\documentclass[conference]{IEEEtran}
\usepackage{cite}
\usepackage{amsmath,amssymb,amsfonts}
\usepackage{algorithmic}
\usepackage{graphicx, ifthen}
\usepackage{subfig}
\usepackage{tikz, pgfplots, pgfmath}
\usetikzlibrary{positioning, fit, calc, decorations.markings}
\usepackage{textcomp}
\usepackage{enumerate}
\usepackage{xcolor}
\usepackage{braket}
\usepackage{physics}
\usepackage{comment}
\usepackage{mathtools}
\usepackage{hyperref}
\usetikzlibrary{3d}
\usepackage{tikz-3dplot}

\usepackage[stretch=15,shrink=20]{microtype}

\def\BibTeX{{\rm B\kern-.05em{\sc i\kern-.025em b}\kern-.08em
    T\kern-.1667em\lower.7ex\hbox{E}\kern-.125emX}}

\usepackage[utf8]{inputenc}
\usepackage[english]{babel}
\usepackage[T1]{fontenc}

\usepackage{lipsum}

\usepackage{makecell}
\setcellgapes{4pt}

\usepackage{algorithm}

\usepackage{amsthm}
\theoremstyle{plain}

\newtheorem{proposition}{Proposition}

\theoremstyle{definition}
\newtheorem{definition}{Definition}

\IEEEoverridecommandlockouts

\begin{document}

\title{Task Concurrency and Compatibility in Measurement-Based Quantum Networks}

\author{\IEEEauthorblockN{Jakob Kaltoft Søndergaard}
\IEEEauthorblockA{\textit{Department of Electronic Systems} \\
\textit{Aalborg University}\\
Aalborg, Denmark \\
jakobks@es.aau.dk}
\and
\IEEEauthorblockN{René Bødker Christensen}
\IEEEauthorblockA{\textit{Department of Mathematical Sciences} \\
\textit{Aalborg University}\\
Aalborg, Denmark \\
rene@math.aau.dk}
\and
\IEEEauthorblockN{Petar Popovski}
\IEEEauthorblockA{\textit{Department of Electronic Systems} \\
\textit{Aalborg University}\\
Aalborg, Denmark \\
petarp@es.aau.dk}
\thanks{This work was supported, in part, by the Danish National Research Foundation (DNRF), through the Center CLASSIQUE, grant nr. 187.}
}

\maketitle

\begin{abstract}
    Measurement-Based Quantum Networks (MBQNs) rely on multipartite pre-shared entanglement resources to satisfy entanglement requests. Traditional designs optimize these resources for individual tasks, neglecting that multiple tasks may arrive concurrently and compete for the same entanglement. We introduce \emph{compatibility} as a design-level metric, capturing whether concurrent tasks can be satisfied by the same entanglement resources. We define a worst-case notion of compatibility where nodes are prevented from coordinating after task arrival and illustrate why tasks may be incompatible. Furthermore, we explore compatibility extensions that account for stochastic arrivals and the capability to supplement the pre-shared entanglement with additional entanglement on-demand, and show that incompatibility differs structurally dependent on the set of concurrent tasks. We argue that compatibility should be used for resource state design, building the foundation for determining which task pairs the network should support with pre-shared entanglement and which require execution-time coordination. Numerical simulations demonstrate this potential, with $(G,1)$-compatibility achieving a 40\%-55\% gain in simultaneously supported tasks relative to the single-task baseline. By incorporating compatibility as a fundamental design objective, quantum networks can move beyond single-task optimization towards scalable, robust architectures that effectively balance proactive entanglement distribution and supplemental reactive coordination.
\end{abstract}


\section{Introduction}
Quantum networks enable fundamentally new communication and distributed computing capabilities by allowing distant nodes to share and manipulate quantum states \cite{wehner2018quantum}. Central to these networks is shared entanglement, a quantum resource that is costly to generate, time-sensitive to maintain, and constrained by how it is distributed and consumed within the network. Consequently, the structure and entanglement topology of a quantum network play a decisive role in determining which tasks can be supported and whether they can be executed concurrently.

Measurement-based quantum networking (MBQN) is a promising architecture, in which it is assumed that multipartite entanglement resources are provisioned ahead of demand, forming a shared entanglement layer that must later be transformed to satisfy tasks \cite{pirker2019quantum}. In this setting, the network must decide in advance which entanglement to generate, distribute, and maintain, without knowing the exact nature of the tasks that will arrive and when they will occur. While entanglement can in principle be generated on demand once a task is known, doing so typically incurs substantial latency and coordination overhead, making proactive entanglement distribution an attractive, but non-trivial, design strategy. The resulting pre-distributed entanglement topology therefore plays a crucial role in how well the network can support not only individual tasks, but also multiple tasks that may arrive concurrently.

Quantum networking is still in its infancy, and hence quantum network architectures are still commonly designed and evaluated through single-task metrics despite the shared nature of entanglement resources. Existing performance metrics typically include maximizing the fidelity or minimizing either latency or storage requirements of an end-to-end entangled connection under the assumption that other demands are absent or can be treated independently \cite{mor2025imperfect, miguel2023optimized, sondergaard2025satellite}. In contrast, multi-task strategies are well established in classical networking architectures, where concurrent traffic can often be managed through buffering, repetition, or mature access-control mechanism. In quantum networks, however, the combination of pre-distributed, consumable entanglement and limited time for coordination at task arrival makes concurrency fundamentally more delicate. Consequently, the question of how multipartite entanglement should be designed to remain effective under simultaneous and potentially competing task demands remains largely open.

\begin{figure*}[t] 
    \centering
    \subfloat[\label{fig: invited__intro_GHZ}]{%
        \tikzset{
        plane/.style={fill=gray!40, opacity=.15, draw=gray!90, line width=.5pt},
        node/.style={circle, fill=white, draw=black, line width=0.3pt, inner sep=1.5pt, font=\sffamily},
        edge/.style={black, thick}
}

\tdplotsetmaincoords{80}{120}

\begin{tikzpicture}[tdplot_main_coords, scale=2]
        
    \begin{scope}[canvas is xy plane at z=0]
        \draw[plane] (-0.3,-0.3) rectangle (1.7,1.4);
        \node[node] at (0,0) (1) {1};
        \node[node] at (1.4,0) (2) {2};
        \node[node] at (0.7,1.1) (3) {3};
        \draw[edge] (1) -- (2);
        \draw[edge] (1) -- (3);
        \draw[edge] (2) -- (3);
    \end{scope}
\end{tikzpicture}}
    \subfloat[\label{fig: invited_intro_EPR}]{%
        \tikzset{
    plane/.style={fill=blue!5, opacity=0.3, draw=blue!20, line width=0.5pt},
    node/.style={circle, fill=white, draw=blue!70!black, text=blue!70!black, line width=0.3pt, inner sep=1.5pt, font=\sffamily},
    node_unused/.style={circle, fill=white, draw=gray!20, text=gray!25, line width=0.2pt, inner sep=1.2pt, font=\sffamily},
    edge/.style={blue!70!black, thick},
    pillar/.style={blue!15, densely dotted, line width=0.3pt}
}

\tdplotsetmaincoords{80}{120}
\begin{tikzpicture}[tdplot_main_coords, scale=2]

    \foreach \x/\y in {0/0, 1.4/0, 0.7/1.1} {
        \draw[pillar] (\x,\y,0) -- (\x,\y,1.0);
    }

    \begin{scope}[canvas is xy plane at z=0]
        \draw[plane] (-0.3,-0.3) rectangle (1.7,1.4);
        \node[node_unused] at (0,0) (01) {1};
        \node[node] at (1.4,0) (02) {2};
        \node[node] at (0.7,1.1) (03) {3};
        \draw[edge] (02) -- (03);
    \end{scope}

    \begin{scope}[canvas is xy plane at z=0.5]
        \draw[plane] (-0.3,-0.3) rectangle (1.7,1.4);
        \node[node] at (0,0) (11) {1};
        \node[node_unused] at (1.4,0) (12) {2};
        \node[node] at (0.7,1.1) (13) {3};
        \draw[edge] (11) -- (13);
    \end{scope}

    \begin{scope}[canvas is xy plane at z=1.0]
        \draw[plane] (-0.3,-0.3) rectangle (1.7,1.4);
        \node[node] at (0,0) (31) {1};
        \node[node] at (1.4,0) (32) {2};
        \node[node_unused] at (0.7,1.1) (33) {3};
        \draw[edge] (31) -- (32); 
    \end{scope}
\end{tikzpicture}}
    \subfloat[\label{fig: invited_intro_EPRs}]{%
        \tikzset{
    plane/.style={fill=red!5, opacity=0.3, draw=red!20, line width=0.5pt},
    node/.style={circle, fill=white, draw=red!70!black, text=red!70!black, line width=0.3pt, inner sep=1.5pt, font=\sffamily},
    edge/.style={red!70!black, thick},
    pillar/.style={red!15, densely dotted, line width=0.3pt}
}

\tdplotsetmaincoords{80}{120}
\begin{tikzpicture}[tdplot_main_coords, scale=2]

    \foreach \x/\y in {0/0, 1.4/0, 0.7/1.1} {
        \draw[pillar] (\x,\y,0) -- (\x,\y,1.0);
    }

    \begin{scope}[canvas is xy plane at z=0]
        \draw[plane] (-0.3,-0.3) rectangle (1.7,1.4);
        \node[node] at (0,0) (01) {1};
        \node[node] at (1.4,0) (02) {2};
        \node[node] at (0.65,1.05) (03A) {3};
        \draw[edge] (01) -- (03A);
        \node[node] at (0.75,1.05) (03B) {3};
        \draw[edge] (02) -- (03B);
    \end{scope}

    \begin{scope}[canvas is xy plane at z=0.5]
        \draw[plane] (-0.3,-0.3) rectangle (1.7,1.4);
        \node[node] at (0,0) (11) {1};
        \node[node] at (1.35,0) (12A) {2};
        \draw[edge] (11) -- (12A);
        \node[node] at (1.45,0) (12B) {2};
        \node[node] at (0.7,1.1) (13) {3};
        \draw[edge] (13) -- (12B);
    \end{scope}

    \begin{scope}[canvas is xy plane at z=1.0]
        \draw[plane] (-0.3,-0.3) rectangle (1.7,1.4);
        \node[node] at (.05,0) (31A) {1};
        \node[node] at (1.4,0) (32) {2};
        \draw[edge] (32) -- (31A);
        \node[node] at (-.05,0) (31B) {1};
        \node[node] at (0.7,1.1) (33) {3};
        \draw[edge] (33) -- (31B); 
    \end{scope}
\end{tikzpicture}}
    \caption{Minimal example illustrating limitations of single-task metrics under concurrent tasks. a) Tripartite entanglement resource shared by the network prior to task arrival. b) Bipartite entanglement between any node pair in isolation is satisfiable by the resource. c) Concurrent independent bipartite entanglements are not compatible by the resource.}
\end{figure*}
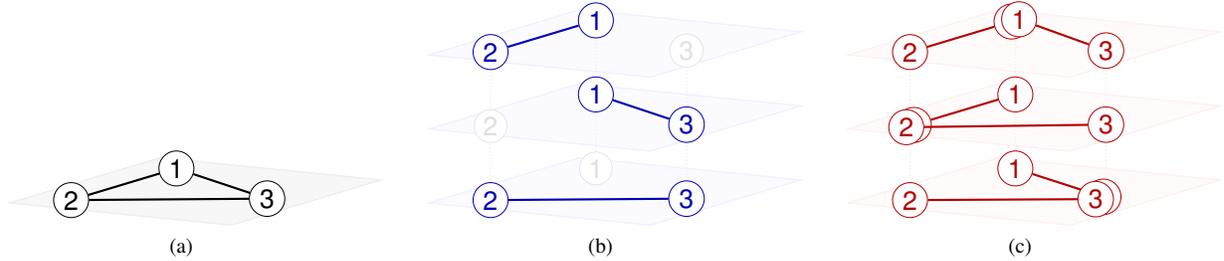

As quantum networks evolve toward large-scale, multi-user networks, it is natural to expect that entanglement requests will arrive in a dynamical, asynchronous, and stochastic manner, rather than being centrally coordinated. Users may request entanglement independently with little advance knowledge of the current resource state or of other simultaneous demands, resulting in random-access-like arrivals. When multiple tasks arrive concurrently, they may compete for the same pre-distributed entanglement. Importantly, not all such concurrency is inherently problematic; some tasks can be satisfied simultaneously whereas others impose conflicting requirements on shared nodes or links \cite{freund2024flexible}. This distinction suggests that beyond the total amount of entanglement available, the \emph{compatibility} of tasks with respect to a given entanglement resource becomes a central measure of network performance under concurrent demand. Under such concurrent arrivals, an entanglement resource that is sufficient for every task in isolation can become jointly infeasible. Two tasks that could each be completed successfully on their own may interfere when executed simultaneously, competing for overlapping portions of the shared entanglement and preventing either from being realized. In this sense, concurrency can lead not merely to performance degradation, but to structural failure; the network may be unable to satisfy any of the active demands despite having been optimized under single-task metrics.

A simple example illustrates this point. Consider three network nodes that share a single tripartite entangled state, such as a GHZ-state (see Fig.~\ref{fig: invited__intro_GHZ}). This state is sufficient to establish an EPR pair between any node pair via local operations, and is therefore effective for satisfying one end-to-end entanglement request at a time (Fig.~\ref{fig: invited_intro_EPR}). However, if two node pairs request entanglement simultaneously, the same resource is fundamentally insufficient; since only three qubits are stored, the network cannot support two independent EPR pairs concurrently (Fig.~\ref{fig: invited_intro_EPRs}). Thus, although the resource is optimal in a single-task sense, it fails immediately under concurrent demand. This minimal scenario highlights that quantum network design must account not only for single-task capability, but also for the compatibility of multiple concurrent tasks.

One way to mitigate such conflicts is to introduce explicit coordination mechanisms that regulate when users are allowed to request tasks, analogous to medium-access control in classical networks \cite{pirker2025resource, illiano2023quantum}. However, such approaches require additional signaling, consume scarce resources, and introduce latency; costs that are particularly pronounced in quantum networks. An alternative is to incorporate compatibility directly into the design of the entanglement topology itself, enabling certain classes of concurrent tasks to be supported without requiring costly coordination at task arrival.

In this paper, we argue that \emph{task compatibility} should be treated as a first-class design objective in the design of entanglement resources in MBQN. We introduce a minimal notion of compatibility that captures whether multiple tasks can be satisfied concurrently by a given pre-distributed entanglement topology, focusing on a worst-case setting with minimal coordination at task arrival. We then extend this framework to \emph{beyond-worst-case compatibility}, allowing for limited execution-time coordination and adaptive entanglement distribution. The purpose is not to present a fully optimized compatibility-aware design framework, but rather to expose the limitations of single-task-centric metrics, reinterpret existing entanglement distribution strategies through the lens of concurrent demand, and outline a research agenda for quantum network designs that remain robust when tasks arrive unpredictably and overlap in time. To illustrate the practical impact of these ideas, we perform numerical simulations, demonstrating substantial and consistent gains in the number of concurrently supported tasks even under minimal coordination.

\section{Quantum Networking Architecture and Resource Model}
In this section, we give a brief introduction to the MBQN framework and present the resource model considered in this paper.

\subsection{Entanglement Resource Graph}
In MBQN, the entanglement resource is represented by a graph state \cite{hein2006entanglement}. Let $G=(V,E)$ be a simple graph where $V$ denotes the set of qubits and $E$ the set of entangling links. The corresponding multipartite graph state $\ket{G}$ is defined as the stabilizer state of the operators
\begin{equation}
    S_v=X_v\prod_{u\in N(v)}Z_u,
    \quad\forall\, v\in V,
\end{equation}
where $X_v,Z_u$ represents the Pauli $X,Z$ operator on qubit $v,u$, and $N(v)$ the neighborhood of node $v$ in $G$.

The resource state $\ket{G}$ is established through an entanglement distribution strategy and is assumed to be available when tasks arrive. Consequently, we do not explicitly focus on the distribution of the resource state. Instead, we focus on the performance of resource states under concurrent tasks. The key architectural question is therefore not only how much entanglement is distributed, but how it is structured across the resource state $\ket{G}$. Figure~\ref{fig: Invited_Network_Setup} illustrates a seven node network sharing a 1D-cluster resource state, which we will use as the setup throughout this paper.

\begin{figure}
    \centering
    \tikzset{network/.style={rectangle, draw, inner sep=6pt},
qubit/.style={circle, black, draw, fill, very thin, inner sep=3pt},
user/.style={rectangle, draw, black, inner sep=7pt}
}

\begin{tikzpicture}[scale=.75, transform shape]

    \node[qubit] at (0,0) (1) {};
    \node[above left=0pt of 1, yshift=-2pt] {$1$};
    \node[user, fit=(1)] (U1) {};

    \foreach \n in {2,...,7}{
        \pgfmathtruncatemacro{\nprior}{\n-1}
        \node[qubit, right=1cm of \nprior] (\n) {};
        \node[above left=0pt of \n, yshift=-2pt] {$\n$};
        \node[user, fit=(\n)] (U\n) {};
        \draw[black] (\nprior) -- (\n);
    }

\end{tikzpicture}
    \caption{Seven node 1D-cluster resource state used throughout the paper.}
    \label{fig: Invited_Network_Setup}
\end{figure}
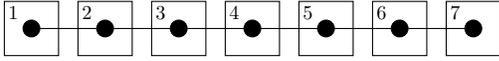

\subsection{Allowed Operations}
Tasks are executed on the shared entanglement by locally transforming the resource state into requested entanglement using a measurement-based approach. Concretely, we assume that nodes are restricted to local operations and classical communication (LOCC). Notably, nodes are allowed to perform Pauli measurements in the $Z,Y$-bases, respectively corresponding to deleting the measured node and its incident edges, and a local complementation of the measured node before deleting it and its incident edges. The graphical action of these measurements can be illustrated by considering their action on the resource state in Fig.~\ref{fig: Invited_Network_Setup} when node 4 is measured: a $Z_4$-measurement deletes node 4 and its edges, thereby separating the graph in two, yielding the separable state $\ket{G_{123}}\otimes\ket{G_{567}}$; a $Y_4$-measurement adds the edge $(3,5)$ before deleting node 4 and its incident edges, resulting in one entangled graph state still in a 1D-cluster.

Importantly, the measurement-based approach is inherently consumptive; once a node is measured to realize a task, the corresponding qubit and its entanglement resources (entangling links) is no longer available for other concurrent tasks.

\subsection{Tasks}
A task corresponds to a request for entanglement. In this work, we consider a network supporting distribution of EPR pairs between any two nodes in the network. In the graph state formalism, an EPR pair between nodes $u,v$ is equivalent to a component consisting of those two nodes, i.e., a connected subgraph with nodes $u,v$ that are not connected to any other nodes.

A task $T=(u,v)$ is said to be feasible by the resource state $\ket{G}$ if the requested entanglement an be obtained from the resource state solely through LOCC. Note that checking feasibility is in general NP-Complete \cite{dahlberg2020transforming}. The simplest approach to satisfy a task is the repeater-path protocol \cite{hahn2019quantum}, which consists of isolating a path $P=(u,\ldots,v)$ by performing $Z$-measurements on the neighbors of the path, followed by $Y$-measurements on the intermediate nodes on the path. Hence, feasibility of $T$ is guaranteed by the existence of a path $P$ with end nodes $u,v$ in $G$. We intentionally restrict attention to the simple repeater-path protocol in order to highlight compatibility as a consequence of architectural limitations of the entanglement resource, rather than of protocol-specific optimization. While more sophisticated protocols may exist \cite{freund2024flexible}, such approaches may improve performance under concurrent demand but do not eliminate the need for a unified compatibility framework.

We assume that tasks arrive stochastically and without prior coordination. Given a set of simultaneously active tasks
\begin{equation}
    \mathcal{T}=\{T_1,\ldots,T_n\},
\end{equation}
each task attempts to consume entanglement from the same underlying resource state. While task $T_i$ may be feasible in isolation, the network may be unable to satisfy all tasks jointly due to the overlapping resource requirements. For simplicity, we primarily consider two concurrent tasks in this work.

The above model captures how individual tasks may be executed from a pre-shared entanglement resource in MBQN. In the next section, we build on this framework by considering the setting with concurrent tasks.

\section{Task Compatibility}
Having introduced the resource and task model, we now turn to the question motivating this work; how should quantum network resources be designed and evaluated when tasks may arrive concurrently and compete for the same pre-shared entanglement. To capture when such tasks can be satisfied simultaneously, we introduce the notion of compatibility.

In this section, we first introduce a worst-case definition of compatibility and illustrate its implications through several concurrent task scenarios. We then move beyond the worst-case setting by discussing more permissive notions of compatibility and how such metrics can be incorporated as a design principle for scalable quantum networks.

\subsection{Worst-Case Compatibility}
To avoid mistaking compatibility with task-arrival coordination, we begin with a minimal worst-case setting in which no additional scheduling or coordination is assumed. While incompatible tasks could in principle be resolved through MAC-like protocols, such approaches induce additional overhead and simply move the burden from the design phase to the execution phase. Instead, we ask which tasks are inherently compatible given only the pre-shared resources and LOCC. Hence, worst-case compatibility isolates concurrency that is enabled by resource state design rather than by execution-time coordination.

\begin{definition}[Worst-Case Task Compatibility]\label{def: invited_worst-case_compatibility}
    Let two tasks $T_1=(u_1,v_1)$ and $T_2=(u_2,v_2)$ be requested from the resource state $G$. We say that $T_1$ and $T_2$ are $G$-compatible if the following two conditions are satisfied:
    \begin{enumerate}[i)]
        \item \textit{Disjointness:}\\
        There exists paths $P_1=(u_1,\ldots,v_1),P_2=(u_2,\ldots,v_2)$ in $G$ such that the paths are vertex-disjoint, i.e.,
        \begin{equation}
            V(P_1)\cap V(P_2)=\emptyset.
        \end{equation}
        \item \textit{Separability:}\\
        The paths $P_1,P_2$ are non-adjacent in the sense that no nodes in $P_1$ neighbor one in $P_2$, i.e.,
        \begin{equation}
            \text{dist}(V(P_1),V(P_2))\geq2.
        \end{equation}
    \end{enumerate}
\end{definition}
When the resource state is inferred, we omit the $G$ and simply say that tasks are compatible. The conditions follows immediately from the required measurements in the repeater-path protocol. Disjointness is a familiar concurrency constraint, where simultaneous tasks must avoid direct contention for the same physical channel resources, e.g., separated in space, time, or frequency. The separability condition, however, is inherently quantum arising from the need to separate entangled subsystems.

\begin{proposition}[Necessity under LOCC]
    In MBQN relying on LOCC via the repeater-path protocol to extract requested entanglement from the resource state $\ket{G}$, disjointness and separability are necessary conditions for two tasks to be simultaneously satisfiable by $\ket{G}$, that is, $G$-compatible.
\end{proposition}
\begin{IEEEproof}[Proof sketch]
Disjointness is necessary because two tasks cannot consume or reserve the same qubit(s) at the same time, thereby implying that they cannot both use it in the repeater-path protocol. Separability is necessary because adjacent paths remain entangled and separating them would otherwise require measuring qubits that are needed for one of the tasks.
\end{IEEEproof}

While these constraints are implicit in the MBQN literature, our contribution is to elevate them into an explicit compatibility criterion for support of concurrent tasks in quantum networks.

A note on the definition of compatibility is in order. This definition represents a worst-case notion of compatibility under no task-arrival coordination; tasks must not only avoid overlapping resources, but must also be sufficiently separated in the network to prevent adjacency-induced conflicts between end-nodes. More permissive notions of compatibility may be appropriate when additional coordination or task-aware scheduling is allowed. This is discussed in Sec.~\ref{sec: Invited_Beyond_Worst-case_Compatibility}. However, this worst-case notion provides a simple baseline for evaluating whether an entanglement distribution strategy supports concurrency beyond the single-task regime.

\subsection{Compatibility Scenarios}
We will now illustrate the worst-case compatibility through four scenarios shown in Fig.~\ref{fig: Invited_cases}. These examples highlight how concurrent tasks are naturally limited by the topology of the pre-shared entanglement resources. In all cases, we consider the seven node 1D-cluster in Fig.~\ref{fig: Invited_Network_Setup} as the resource state.

\begin{figure}[t] 
    \centering
    \subfloat[\label{fig: invited__cases_cover}]{%
        \tikzset{network/.style={rectangle, draw, inner sep=6pt},
qubit/.style={circle, gray, draw, fill, very thin, inner sep=3pt},
user/.style={rectangle, draw, gray, inner sep=8pt}
}

\begin{tikzpicture}[scale=.75, transform shape]

    \node[qubit] at (0,0) (1) {};
    \node[above left=0pt of 1, yshift=-.5pt] {\textcolor{gray}{$1$}};
    \node[user, fit=(1)] (U1) {};

    \node[qubit, right=1cm of 1] (2) {};
    \node[above left=0pt of 2, yshift=-.5pt] {\textcolor{gray}{$2$}};
    \node[user, fit=(2)] (U2) {};

    \node[qubit, right=1cm of 2] (3) {};
    \node[above left=0pt of 3, yshift=-.5pt] {\textcolor{gray}{$3$}};
    \node[user, fit=(3)] (U3) {};

    \node[qubit, right=1cm of 3] (4) {};
    \node[above left=0pt of 4, yshift=-.5pt] {\textcolor{gray}{$4$}};
    \node[user, fit=(4)] (U4) {};

    \node[qubit, right=1cm of 4] (5) {};
    \node[above left=0pt of 5, yshift=-.5pt] {\textcolor{gray}{$5$}};
    \node[user, fit=(5)] (U5) {};

    \node[qubit, right=1cm of 5] (6) {};
    \node[above left=0pt of 6, yshift=-.5pt] {\textcolor{gray}{$6$}};
    \node[user, fit=(6)] (U6) {};

    \node[qubit, right=1cm of 6] (7) {};
    \node[above left=0pt of 7, yshift=-.5pt] {\textcolor{gray}{$7$}};
    \node[user, fit=(7)] (U7) {};

    \draw[gray] (1) -- (7);
    \draw[thick, black] (3.north) -- (4.north);
    \draw[dashed, thick, black] (1.south) -- (6.south);
\end{tikzpicture}}
    \hspace{1cm }
    \subfloat[\label{fig: invited__cases_overlap}]{%
        \tikzset{network/.style={rectangle, draw, inner sep=6pt},
qubit/.style={circle, gray, draw, fill, very thin, inner sep=3pt},
user/.style={rectangle, draw, gray, inner sep=8pt}
}

\begin{tikzpicture}[scale=.75, transform shape]

    \node[qubit] at (0,0) (1) {};
    \node[above left=0pt of 1, yshift=-.5pt] {\textcolor{gray}{$1$}};
    \node[user, fit=(1)] (U1) {};

    \node[qubit, right=1cm of 1] (2) {};
    \node[above left=0pt of 2, yshift=-.5pt] {\textcolor{gray}{$2$}};
    \node[user, fit=(2)] (U2) {};

    \node[qubit, right=1cm of 2] (3) {};
    \node[above left=0pt of 3, yshift=-.5pt] {\textcolor{gray}{$3$}};
    \node[user, fit=(3)] (U3) {};

    \node[qubit, right=1cm of 3] (4) {};
    \node[above left=0pt of 4, yshift=-.5pt] {\textcolor{gray}{$4$}};
    \node[user, fit=(4)] (U4) {};

    \node[qubit, right=1cm of 4] (5) {};
    \node[above left=0pt of 5, yshift=-.5pt] {\textcolor{gray}{$5$}};
    \node[user, fit=(5)] (U5) {};

    \node[qubit, right=1cm of 5] (6) {};
    \node[above left=0pt of 6, yshift=-.5pt] {\textcolor{gray}{$6$}};
    \node[user, fit=(6)] (U6) {};

    \node[qubit, right=1cm of 6] (7) {};
    \node[above left=0pt of 7, yshift=-.5pt] {\textcolor{gray}{$7$}};
    \node[user, fit=(7)] (U7) {};

    \draw[gray] (1) -- (7);
    \draw[thick, black] (2.north) -- (6.north);
    \draw[dashed, thick, black] (4.south) -- (7.south);
\end{tikzpicture}}
    \\    
    \subfloat[\label{fig: invited_cases_disjoint}]{%
        \tikzset{network/.style={rectangle, draw, inner sep=6pt},
qubit/.style={circle, gray, draw, fill, very thin, inner sep=3pt},
user/.style={rectangle, draw, gray, inner sep=8pt}
}

\begin{tikzpicture}[scale=.75, transform shape]

    \node[qubit] at (0,0) (1) {};
    \node[above left=0pt of 1, yshift=-.5pt] {\textcolor{gray}{$1$}};
    \node[user, fit=(1)] (U1) {};

    \node[qubit, right=1cm of 1] (2) {};
    \node[above left=0pt of 2, yshift=-.5pt] {\textcolor{gray}{$2$}};
    \node[user, fit=(2)] (U2) {};

    \node[qubit, right=1cm of 2] (3) {};
    \node[above left=0pt of 3, yshift=-.5pt] {\textcolor{gray}{$3$}};
    \node[user, fit=(3)] (U3) {};

    \node[qubit, right=1cm of 3] (4) {};
    \node[above left=0pt of 4, yshift=-.5pt] {\textcolor{gray}{$4$}};
    \node[user, fit=(4)] (U4) {};

    \node[qubit, right=1cm of 4] (5) {};
    \node[above left=0pt of 5, yshift=-.5pt] {\textcolor{gray}{$5$}};
    \node[user, fit=(5)] (U5) {};

    \node[qubit, right=1cm of 5] (6) {};
    \node[above left=0pt of 6, yshift=-.5pt] {\textcolor{gray}{$6$}};
    \node[user, fit=(6)] (U6) {};

    \node[qubit, right=1cm of 6] (7) {};
    \node[above left=0pt of 7, yshift=-.5pt] {\textcolor{gray}{$7$}};
    \node[user, fit=(7)] (U7) {};

    \draw[gray] (1) -- (7);
    \draw[thick, black] (1.north) -- (3.north);
    \draw[dashed, thick, black] (4.south) -- (6.south);
\end{tikzpicture}}
    \hspace{1cm}    
    \subfloat[\label{fig: invited_cases_compatible}]{%
        \tikzset{network/.style={rectangle, draw, inner sep=6pt},
qubit/.style={circle, gray, draw, fill, very thin, inner sep=3pt},
user/.style={rectangle, draw, gray, inner sep=8pt}
}

\begin{tikzpicture}[scale=.75, transform shape]

    \node[qubit] at (0,0) (1) {};
    \node[above left=0pt of 1, yshift=-.5pt] {\textcolor{gray}{$1$}};
    \node[user, fit=(1)] (U1) {};

    \node[qubit, right=1cm of 1] (2) {};
    \node[above left=0pt of 2, yshift=-.5pt] {\textcolor{gray}{$2$}};
    \node[user, fit=(2)] (U2) {};

    \node[qubit, right=1cm of 2] (3) {};
    \node[above left=0pt of 3, yshift=-.5pt] {\textcolor{gray}{$3$}};
    \node[user, fit=(3)] (U3) {};

    \node[qubit, right=1cm of 3] (4) {};
    \node[above left=0pt of 4, yshift=-.5pt] {\textcolor{gray}{$4$}};
    \node[user, fit=(4)] (U4) {};

    \node[qubit, right=1cm of 4] (5) {};
    \node[above left=0pt of 5, yshift=-.5pt] {\textcolor{gray}{$5$}};
    \node[user, fit=(5)] (U5) {};

    \node[qubit, right=1cm of 5] (6) {};
    \node[above left=0pt of 6, yshift=-.5pt] {\textcolor{gray}{$6$}};
    \node[user, fit=(6)] (U6) {};

    \node[qubit, right=1cm of 6] (7) {};
    \node[above left=0pt of 7, yshift=-.5pt] {\textcolor{gray}{$7$}};
    \node[user, fit=(7)] (U7) {};

    \draw[gray] (1) -- (7);
    \draw[thick, black] (1.north) -- (3.north);
    \draw[dashed, thick, black] (5.south) -- (6.south);
\end{tikzpicture}}
    \caption{Compatibility examples on the resource state of Fig.~\ref{fig: Invited_Network_Setup}. Highlighted paths represent concurrent tasks requests, solid ($T_1$) and dashed ($T_2$). a) Covering tasks. b) Intersecting tasks. c) Disjoint but adjacent tasks. d) Separated tasks.}
    \label{fig: Invited_cases}
\end{figure}
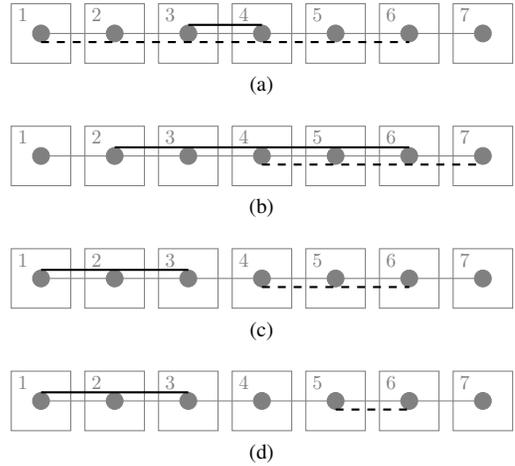

\subsubsection{Covering Tasks}
Consider two tasks $T_1=(3,4)$ and $T_2=(1,6)$. Then $P_1=(3,4)$ and $P_2=(1,2,3,4,5,6)$ implying that $P_2$ covers $P_1$, i.e., $P_1\subsetneq P_2$. Consequently, it immediately follows that the tasks are incompatible under Def.~\ref{def: invited_worst-case_compatibility}; to satisfy $T_2$, all intermediate nodes in $P_2$ are measured and therefore unentangled from the remaining entanglement resource. Particularly, nodes~3 and 4 are measured and separated implying that no local operations can entangle them to satisfy $T_1$. Note furthermore that nodes~2 and 5 would need to use two different measurement bases to satisfy the tasks: to satisfy $T_1$, both nodes measure in the $Z$-basis to separate the subgraph $(3,4)$ from the rest of the graph; to satisfy $T_2$, both nodes measure in the $Y$-basis to obtain the edge $(1,6)$. The arrival times of feedforward information at these nodes are therefore crucial for the tasks, which we further discuss in Sec.~\ref{sec: invited_partial_compatibility_timing}.

\subsubsection{Intersecting Tasks}
Consider two tasks $T_1=(2,6)$ and $T_2=(4,7)$. In this particular case, the tasks are intersecting with $V(P_1)\cap V(P_2)=\{4,5,6\}$. Once again, the tasks are incompatible due to being non-disjoint such that intermediate nodes on both paths must measure. Note that node 5 must measure in the $Y$-bases for both tasks. However, node 4 must measure to satisfy $T_1$ making $T_2$ unsatisfiable, and node 6 must measure to satisfy $T_2$ making $T_1$ unsatisfiable. As a result, these tasks are fighting for the same resources in a way that cannot be resolved with the available resources without either of the tasks being sacrificed to satisfy the other.

\subsubsection{Disjoint but Adjacent Tasks}
Now, consider two disjoint tasks $T_1=(1,3)$ and $T_2=(4,6)$ that has two adjacent end-nodes. In contrast to classical routing, disjointness of resources alone does not guarantee compatible tasks. As the two tasks are entangled through the edge $(3,4)$, separating the tasks requires either node 3 or 4 to measure in the $Z$-basis. Consequently, the tasks cannot be satisfied simultaneously.

\subsubsection{Separated Tasks}
Consider the case where the two disjoint tasks $T_1=(1,3)$ and $T_2=(5,6)$ are separated by an intermediate node, 4. With this separation, the two tasks can be separated by node 4 performing a $Z$-measurement, which was not possible before when the tasks were adjacent. In this case, both conditions in Def.~\ref{def: invited_worst-case_compatibility} are satisfied and the two tasks are therefore compatible.

\subsection{Beyond Worst-Case Compatibility}\label{sec: Invited_Beyond_Worst-case_Compatibility}
The worst-case notion of task compatibility introduced above assumes minimal coordination and relies solely on a fixed pre-distributed entanglement resource. It therefore provides a conservative baseline; two tasks are compatible only if they can be satisfied concurrently without any execution-time signaling or assistance. Realistic quantum networks, however, are not restricted to only a fixed resource topology and LOCC, but may design richer resource states, consider timing effects, or even distribute additional entanglement to supplement the pre-shared entanglement with the intention of increasing the likelihood of concurrent tasks being compatible. In this section, we briefly discuss three considerations that move beyond worst-case compatibility and illustrate how these considerations relate to the four compatibility scenarios from Fig.~\ref{fig: Invited_cases}.

\subsubsection{Resource State Redesign}
The first approach for making $G$-incompatible tasks compatible is to redesign the topology of $G$. Consider the two tasks in Fig.~\ref{fig: invited__cases_cover}, $T_1=(3,4),T_2=(1,6)$, which are incompatible due to $T_2$ covering $T_1$. A simple solution to this is to add the edge $(1,7)$ to the resource state, thereby forming a ring topology. As a result, $T_2$ can be satisfied by the path $P_2=(1,7,6)$, implying that the two tasks become compatible under Def.~\ref{def: invited_worst-case_compatibility}, as they are now both disjoint and separated. Remarkably, the two intersecting tasks in Fig.~\ref{fig: invited__cases_overlap} remain incompatible, as the overlap creates a ``knot'' that cannot be resolved by such simple rerouting, even though the overlap of these tasks is, in some sense, smaller. Consequently, compatibility is not solely determined by the size of the overlap, but also by its structural configuration.

Adding edges is ``free'' from a storage perspective, since no additional qubits need to be stored. However, the nonlocal operations required to realize such edges during the distribution phase may be too costly or even infeasible. Furthermore, adding edges to support one pair of concurrent tasks can reduce compatibility for other task pairs. In this example, the $G$-compatible tasks $(1,2)$ and $(6,7)$ become adjacent after adding the edge and thereby rendered incompatible.

In general, attempting to design a single resource state that supports all possible combinations of concurrent tasks becomes unscalable as the space of potential demand grows. Consequently, resource state design must balance several system-level considerations, including: (i) the cost of distributing and storing entanglement in terms of latency, communication, and quantum memory overhead; (ii) the expected task arrival statistics, which determine which compatibilities are most valuable to provision for; and (iii) fairness objectives, ensuring that network resources are not optimized solely for a small subset of nodes or tasks. These considerations highlight that tasks incompatibility is not merely binary, but reflects the structural trade-offs inherent to how entanglement is distributed across the network.

Finally, we note that prior works on resource state design for multi-pair entanglement generation adopt a fundamentally different notion of concurrency \cite{miguel2023optimized, sondergaard2025satellite}. In these works, a single task may contain multiple EPR pairs simultaneously, effectively embedding concurrency within the task definition itself. In contrast, we restrict each task to represent only a single EPR pair demand, but allow tasks to arrive sequentially and independently in time. As a result, concurrency in our framework emerges dynamically from the temporal overlap of multiple independent demands, rather than being fixed \emph{a priori}. This distinction leads to qualitatively different resource state design trade-offs, as the network must support unpredictable and temporally evolving concurrency patterns rather than a predetermined set of simultaneous entanglement requests.

\subsubsection{Satisfying Incompatible Tasks Partially}\label{sec: invited_partial_compatibility_timing}
The worst-case compatibility assumes `concurrent' to mean `perfectly simultaneous with no timing structure in their arrivals'. In practice, tasks arrive stochastically, and nodes therefore generally only perform local operations based on partial task information. Particularly, nodes will operate based on the first feedforward information that they receive. This has an important implication: incompatibility of tasks does not necessarily imply that neither task can be satisfied. Instead, when two incompatible tasks arrive within a short time window, the network may complete one task at the cost of leaving insufficient entanglement for executing the other task. Therefore, incompatible tasks resemble a contention scenario in which one request may be served at the expense of another. This is exemplified by all of the incompatible scenarios above. Note that this may induce a bias towards some tasks. For example, in the covering tasks, if either node 2 or 5 measures in the $Z$-basis to execute the inner task, the outer task cannot be satisfied no matter the order of task arrivals. In other hand, if both nodes measures in the $Y$-basis, then both task can be satisfied at the cost of the other nodes sacrificing their entanglement. This highlights the temporal considerations that must be taken into account when handling a multi-task quantum network, and motivates a timing measure of compatibility.

It should be noted that the ability to satisfy one of the concurrent, incompatible tasks is, however, not a guarantee. Consider the intersecting tasks $T_1=(2,6),T_2=(4,7)$ in Fig.~\ref{fig: invited__cases_overlap}, requested by nodes 2 and 7, respectively. That is, task feedforward information flows from these nodes along the paths to the respective end nodes. One could imagine that the tasks arrive in a way such that node 4 receives feedforward information from node 2 before node 7, while node 6 receives from node 7 before node 2. Completely unknowingly of them being end nodes in a requested task, nodes~4 and 6 both measure in the $Y$-basis to satisfy the other task. As a result, both nodes have already sacrificed their entanglement to satisfy one task at the time they learn about them being end nodes in another task. In this case, neither task can be satisfied, which further motivates the notion of tasks being partially compatible, not only relative to the entanglement resource, but also time. Compatibility should be viewed not only as a static property of a resource state, but a dynamic measure intertwined with traffic and timing.

As a preliminary extension to a time-dependent notion of compatibility, one may define $(G,\Delta t)$-\emph{partial compatibility}, which captures whether at least one of the two task can still be successfully satisfied when their arrival times differ by at most $\Delta t$. As the measurement-based approach consumes the shared resources in an order-dependent manner, such a notion is inherently asymmetric; partial compatibility may depend on which task arrives first, as well as on the direction in which task information and local operations propagate through the network. We note that this notion of partial compatibility is fundamentally weaker than the worst-case notion in Def.~\ref{def: invited_worst-case_compatibility}. In the limit $\Delta t\to0$, it reduces to asking whether at least one of the two tasks can be satisfied rather than whether both can be satisfied concurrently.

\subsubsection{Supplemental Entanglement Compatibility}
A third extension is to relax the assumption of zero execution-time coordination. In realistic quantum networks, entanglement distribution will not be restricted to preemptive distribution of the resource state, but may also be permitted upon task arrival, possibly at the cost of additional latency. Allowing such on-demand entanglement distribution to supplement the pre-shared entanglement resource can reduce the destructive interference between concurrent tasks. This motives a quantitative extension of compatibility: rather than asking whether two tasks are compatible strictly from the pre-shared entanglement, one may ask how much additional entanglement must be distributed on-demand to make the tasks compatible. In this view, compatibility becomes a measure of coordination cost, where tasks that only require minor supplemental entanglement are nearly compatible, whereas tasks requiring substantial additional entanglement may be considered incompatible under the network and must be controlled differently.

To ensure that on-demand entanglement is somewhat cheap, we consider how many EPR pairs nodes must distribute to neighboring nodes in order to make both tasks compatible. Consider first the covering tasks in Fig.~\ref{fig: invited__cases_cover}. Clearly, a single EPR pair distributed between nodes 3 and 4 will make these tasks compatible; the inner task is completely satisfied by this distributed entanglement, implying that the entire pre-shared entanglement can be allocated for the covering task. The disjoint but adjacent tasks in Fig.~\ref{fig: invited_cases_disjoint} also suffice with one EPR pair shared among nodes 3 and 4 to become compatible. By using an approach similar to the graph state fission formalism \cite{miguel2025graph}, one can remove the edge $(3,4)$ in the entanglement resource at the cost of an EPR pair shared between these two nodes.\footnote{By denoting the qubits in the distributed EPR pair as $3',4'$, the required operations are $Y_{4'}\text{CZ}_{4,4'}Y_{3'}\text{CZ}_{3,3'}$, where $\text{CZ}_{a,b}$ is the controlled-phase operator on qubits $a,b$ that adds an edge $(a,b)$ if it does not exists, and removes it if it does.} Consequently, the paths induced by the two tasks are separated, thereby making them compatible. Due to the intersecting nature of the two tasks in Fig.~\ref{fig: invited__cases_overlap}, these tasks cannot become compatible by only one additional EPR pair. Instead, the network can distribute three EPR pairs between the pairs (4,5), (5,6), and (6,7), whereafter the path needed to satisfy $T_2$ can be obtained through a sequence of entanglement swappings (or merging measurements followed by $Y$-measurements on the intermediate nodes if one wants to rely solely on the graph state formalism). The idea is, as with the covering tasks, to distribute entanglement such that one task is satisfied by the pre-shared entanglement resources, one by the on-demand distributed entanglement.

Motivated by these examples, a natural extension of compatibility is to incorporate on-demand entanglement distribution. Thus, one may define $(G,k)$-\emph{compatibility}, where $k$ measures the number of additional one-hop EPR pairs that must be distributed on-demand in order for two tasks to become compatible. This definition is a proper generalization of worst-case compatibility; two tasks that are compatible in the worst-case notion are trivially $(G,0)$-compatible. More importantly, $(G,k)$-compatibility captures that incompatibly is not binary in practice, where different incompatible task pairs may require different amounts of reactive entanglement to resolve their concurrency conflicts.

These considerations highlight a more general design question ion MBQN; how much entanglement should be pre-distributed proactively, and how much capacity should be reserved for on-demand entanglement distribution and coordination once the task is known.

\subsection{Compatibility as a Design Metric}
The notions introduced above suggest that compatibility should not be considered as one binary property, but as a design-level objective that complements common single-task performance metrics. Instead of asking whether a network can support a given task in isolation, compatibility asks how well the pre-shared entanglement resource supports sets of concurrent task that generally fight for the same resources.

Importantly, no single notion of compatibility is sufficient in realistic networks: The worst-case compatibility is a conservative baseline in the absence of coordination; time-dependent notions such as $(G,\Delta t)$-partial compatibility captures the important temporal impact of stochastic arrivals; coordination-enabling notions such as $(G,k)$-compatibility quantify how much on-demand entanglement distribution is required to simultaneously accommodate tasks that may otherwise be incompatible. However, in realistic networks, these dimensions are intertwined; tasks arrive stochastically, nodes act based on partial information, some nodes may have capability to distribute limited entanglement on-demand. More precisely, the capability to distribute on-demand entanglement is time-dependent, and whether time-dependent partial compatible tasks are also compatible depends on those capabilities. Consequently, practical compatibility metrics should combine both the temporal structure and entanglement distribution cost to provide a unified framework of compatibility that balances proactive entanglement distribution against reactive coordination.

Overall, compatibility serves as a design-level metric for quantum resource states: it identifies those concurrent demands that should by supported directly by the pre-shared entanglement, and those that must instead be resolved through execution-time coordination.

\section{Numerical Results}
In this section, we investigate how many tasks can be supported concurrently under stochastic demand in a network of $N$ nodes that pre-share a 1D-cluster resource state. Tasks arrive sequentially and independently, modeling a random-access-like demand in a multi-user quantum network. Each task corresponds to an ordered node pair $T_i=(u_i,v_i)$ with $u_i\neq v_i$, drawn uniformly at random from all such pairs. Starting from the empty set, tasks are sequentially appended to the set of simultaneously active tasks $\mathcal{T}$. After each task arrival, we test whether the augmented task set remains compatible under a given compatibility measure. Task arrivals terminate once incompatibility occurs, and the number of compatible tasks accumulated up to this point is recorded.

To isolate architectural concurrency limits and refrain from introducing a temporal task arrival model, we compare three compatibility measures:
\begin{enumerate}
    \item \emph{Baseline}: The simplistic case in which each resource state is responsible for satisfying only a single task.
    \item \emph{Worst-case compatibility}: Tasks must be compatible without any execution-time coordination, i.e., be disjoint and separable.
    \item \emph{$(G,1)$-compatibility}: In addition to the pre-shared entanglement, a single-hop EPR pair may be distributed on-demand to enable resolution of one structural conflict.
\end{enumerate}

For each network size $N$, we perform $10^4$ independent trials. Figure~\ref{fig: compatibility_simulation} shows the average number of simultaneously supported (compatible) tasks together with the corresponding standard error of the mean.

\begin{figure}
    \centering
    \includegraphics[width=\linewidth]{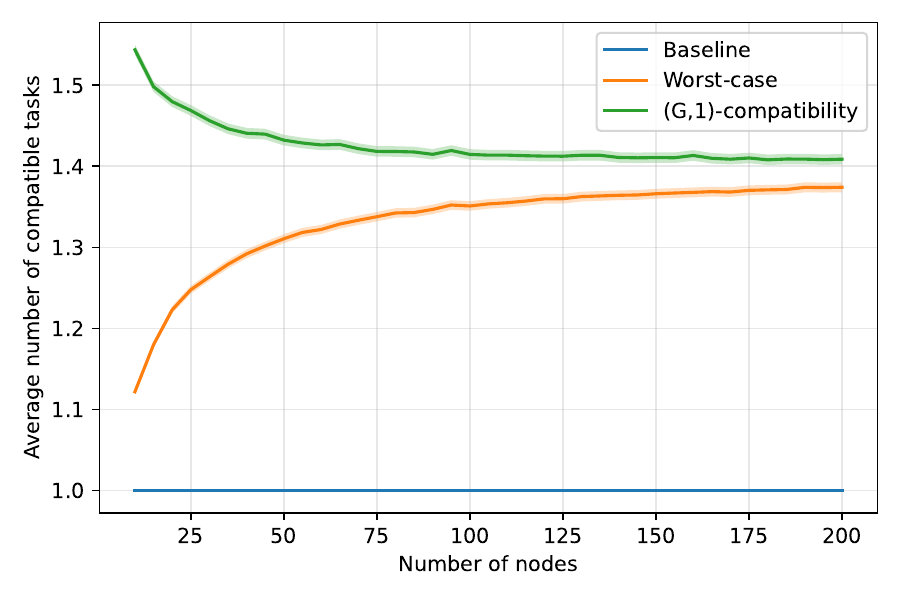}
    \caption{Average number of simultaneously supported tasks as a function of network size for three different compatibility measures: baseline, worst-case compatibility, and $(G,1)$-compatibility. Shaded regions indicate the standard error of the mean over $10^4$ trials.}
    \label{fig: compatibility_simulation}
\end{figure}

Under the baseline measure, each resource state only supports a single task, independent of network size and concurrent demand. In contrast, worst-case compatibility enables multiple tasks to be satisfied by the same resource, with the number of compatible tasks increasing logarithmically with $N$. This demonstrates that even strict architectural constraints permits more efficient entanglement utilization when concurrent demands are taken into account. 

Allowing a single EPR pair on-demand via $(G,1)$-compatibility yields a further gain for all network sizes. Although only one additional EPR pair is permitted, its adaptive placement enables resolution of conflicts arising from tasks that would otherwise require overlapping nodes or adjacent paths, thereby increasing the size of the compatible active task set. As the network size increases, the gain of $(G,1)$-compatibility decreases logarithmically, reflecting the reduced likelihood of such conflicts on larger topologies. Nonetheless, for small networks, even minimal execution-time entanglement distribution can substantially enhance entanglement utilization.

These results highlight three key insights.

First, single-task operation severely underestimates the concurrency potential of quantum networks. Even under strict worst-case assumptions without execution-time coordination, multiple simultaneous entanglement demands can be satisfied.

Second, compatibility is fundamentally an architectural property governed by the topology of the entanglement resource and the permitted local operations, rather than a protocol-specific artifact.

Third, small amounts of adaptive entanglement distribution can mitigate topological conflicts in concurrent tasks execution, suggesting a powerful trade-off between storage overhead, distribution latency, and achievable throughput.

\section{Discussion and Future Work}
The framework introduced in this paper illustrates a structural gap between how entanglement resources are commonly designed and how they are requested and consumed in practice. While the existing literature mainly focuses on optimizing the entanglement resource for individual tasks, realistic quantum networks will operate under stochastic demand, where tasks may arrive almost at the same time and compete for entanglement from the same shared multipartite state. In such networks, single-task optimality does not necessarily translate into scalable multi-task performance. The central contribution of this work is proposing \emph{compatibility} as a missing design dimension. Incorporating the potential demand traffic directly into the design phase is essential to accommodating concurrent demand inherent to multi-user networks.

Several directions remain open. First, we have restricted our attention to tasks being bipartite EPR pairs supported by the repeater-path protocol. In practice, more general application demands involve multipartite states and concurrent requests by the same node. Second, we have not directly considered common single-task metrics such as storage requirements, latency, or end-to-end fidelity. It remains unexplored how these are related when introducing compatibility as another design dimension. Third, incorporating realistic traffic models and arrival processes would allow compatibility to be studied statistically, connecting entanglement provisioning to classical notions of contention and scheduling. Fourth, richer compatibility metrics that jointly account for temporal aspects and reactive entanglement supplementation, e.g., combining $(G,\Delta t)$- and $(G,k)$-type notations, may provide a more complete characterization of performance in realistic networks. Lastly, an important practical question is how to design and distribute pre-shared entanglement resources that are compatibility-aware, given architectural constraints such as limited memory, noisy links and operations, and application specific requirements to the task execution.

\section{Conclusion}
In this paper, we introduced compatibility as a design-level metric for quantum networks relying on pre-shared entanglement resources. We argued that single-task optimization metrics are fundamentally challenged by concurrent demand in realistic quantum networks, since simultaneous demands may contend for the same entanglement resources. To address this, we introduced a minimal worst-case notion of compatibility, illustrated its implications through several concurrent-task scenarios, and discussed compatibility extensions that incorporate resource redesign, timing aspects, and the ability to supplement with additional entanglement on-demand. Simulations on a simplistic network setting furthermore demonstrated how compatibility is a fundamental architectural property, showing gains in the order of 10\%-55\% in simultaneously supporting tasks depending on the given compatibility measure relative to the single-task baseline.

With the introduction of compatibility, concurrency becomes a first-class objective in resource state design. It helps determine which task combinations should be supported by the pre-shared entanglement, and which must instead rely on additional execution-time coordination.

Compatibility-aware resource design may ultimately play a role analogous to multiple-access design in classical networks, forming a foundation for scalable quantum network architectures.

\bibliographystyle{IEEEtran}
\bibliography{bibliography_invited}

\end{document}